\documentclass[11pt,a4paper]{amsart}
\usepackage{fullpage}
\usepackage[foot]{amsaddr}
\usepackage{microtype}
\usepackage[OT1]{fontenc}
\usepackage{eulervm}
\usepackage[tt=false]{libertine} %% tt is ugly
\usepackage{amsmath}
\usepackage{amssymb}
\usepackage{amsthm}
\usepackage{mathtools} % Adds cases* and dcases environments
\usepackage{algorithm}
\usepackage{algpseudocode}
\usepackage{enumitem}
\usepackage[round]{natbib}

\usepackage[margin=1cm]{caption} % Adds an addition margin on either side of figure captions. This must come before "\usepackage{subfig}" (otherwise there are errors).
\usepackage{subfig}

\usepackage[thinlines]{easytable}

\usepackage[bookmarks=true,hypertexnames=false,pagebackref]{hyperref}
\hypersetup{colorlinks=true, citecolor=blue, linkcolor=red, urlcolor=blue}

\usepackage{tikz}
\usetikzlibrary{arrows,arrows.meta,backgrounds,calc,fit,decorations.pathreplacing,decorations.markings,shapes.geometric}

\tikzstyle{internal} = [draw, fill, shape=circle]
\tikzstyle{external} = [shape=circle]
\tikzstyle{square}   = [draw, fill, rectangle]
\tikzstyle{triangle} = [draw, fill, regular polygon, regular polygon sides=3, inner sep=3pt]
\tikzstyle{pentagon} = [draw, fill, regular polygon, regular polygon sides=5, inner sep=2pt, minimum size=14pt]
\tikzset{every fit/.append style=text badly centered}

\usetikzlibrary{positioning,chains,fit,shapes,calc}
\usetikzlibrary{trees}
\usetikzlibrary{decorations.pathreplacing}
\usetikzlibrary{decorations.pathmorphing}
\usetikzlibrary{decorations.markings}
\tikzset{>=latex} % arrow tips
\usepackage{ifthen}

\usepackage{cleveref}

\usepackage[textsize=tiny]{todonotes}

\usepackage[normalem]{ulem}

\usepackage{mleftright}

\usepackage{cool}
\Style{DSymb={\mathrm d},DShorten=true,IntegrateDifferentialDSymb=\mathrm{d}}

\newcommand{\Ex}{\mathop{\mathbb{{}E}}\nolimits}

\def\*#1{\mathbf{#1}}
\def\+#1{\mathcal{#1}}
\def\-#1{\mathrm{#1}}
\def\=#1{\mathbb{#1}}

\newcommand{\abs}[1]{\ensuremath{\left\vert#1\right\vert}}

\newcommand{\RWup}[1]{\ensuremath{P^{\wedge}_{#1}}}
\newcommand{\Pup}[1]{\ensuremath{P^{\uparrow}_{#1}}}
\newcommand{\RWdown}[1]{\ensuremath{P^{\vee}_{#1}}}
\newcommand{\Pdown}[1]{\ensuremath{P^{\downarrow}_{#1}}}

\newcommand{\Diri}[2]{\ensuremath{\+E_{#1}\left(#2\right)}}
\newcommand{\Ent}[2]{\ensuremath{\textnormal{Ent}_{#1}\left(#2\right)}}
\newcommand{\Var}[2]{\ensuremath{\textnormal{Var}_{#1}\left(#2\right)}}

\newcommand{\diag}{\operatorname{diag}}

\newcommand{\transpose}[1]{\ensuremath{#1^{\textnormal{\texttt{T}}}}}

\newcommand{\defeq}{:=}
\newcommand{\identity}{\ensuremath{\*I}}

\newtheorem{theorem}{Theorem}

\newtheorem{lemma}[theorem]{Lemma}

\newtheorem{proposition}[theorem]{Proposition}
\newtheorem{corollary}[theorem]{Corollary}

\crefname{theorem}{Theorem}{Theorems}
\crefname{observation}{Observation}{Observations}
\crefname{claim}{Claim}{Claims}
\crefname{condition}{Condition}{Conditions}
\crefname{algorithm}{Algorithm}{Algorithms}
\crefname{property}{Property}{Properties}
\crefname{example}{Example}{Examples}
\crefname{fact}{Fact}{Facts}
\crefname{lemma}{Lemma}{Lemmas}
\crefname{corollary}{Corollary}{Corollaries}
\crefname{definition}{Definition}{Definitions}
\crefname{remark}{Remark}{Remarks}
\crefname{proposition}{Proposition}{Propositions}
\crefname{equation}{equation}{equations}
\crefname{enumi}{Case}{Case}
\creflabelformat{enumi}{(#2#1#3)}

\makeatletter
\def\prob#1#2#3{\goodbreak\begin{list}{}{\labelwidth\z@ \itemindent-\leftmargin
      \itemsep\z@  \topsep6\p@\@plus6\p@
      \let\makelabel\descriptionlabel}
  \item[\textbf{Name}]#1
  \item[\textbf{Instance}]#2
  \item[\textbf{Output}]#3
  \end{list}}
\makeatother

\makeatletter
\providecommand\@dotsep{5}
\def\listtodoname{Todo list}
\def\listoftodos{\@starttoc{tdo}\listtodoname}
\makeatother

\title{Local-to-Global Contraction in Simplicial Complexes}

\author{Heng Guo}
\author{Giorgos Mousa}

\address{School of Informatics, University of Edinburgh, Informatics Forum, Edinburgh, EH8 9AB, United Kingdom.}
\email{hguo@inf.ed.ac.uk, g.mousa@ed.ac.uk}

\begin{document}

\begin{abstract}
%A further analysis of exchange walks defined over simplicial complexes with an alternative proof of some bounds obtained by \cite{AL20}. Identifying a property of ``local'' walks that leads to modified log-Sobolev inequalities of exchange walks over faces of the same cardinality.
  We give a local-to-global principle for relative entropy contraction in simplicial complexes.
  This is similar to the local-to-global principle for variances obtained by \cite{AL20}.
\end{abstract}

\maketitle

\section{Introduction}

High-dimensional expanders are a powerful new tool for analyzing mixing times of Markov chains which has resulted in many recent successes \citep{ALOV19,CGM21,ALO20,CLV20a,FGYZ20,CGSV20,CLV20}.
The main insight is to recast Markov chains as random walks over faces of simplicial complexes of the same cardinality.
The mixing time of these ``global'' random walks can be bounded via analyzing certain ``local'' walks, thanks to the inductive structure of simplicial complexes.
These local walks are defined over the skeleton of the faces of the simplicial complex, and are typically much easier to analyze.
(Detailed definitions are given in \Cref{sec:prelim}.)

This ``local-to-global'' principle is the key to all of the recent development along this line of research.
Roughly speaking, the local walks contain enough information about the high-dimensional global walks,
yet in the mean time they are much more tractable to analyze.
This methodology turns out to be very powerful \citep{DK17,Opp18,KO20,KM20,LMY20}.
In particular, the result of \cite{KO20} plays a crucial role in resolving the long-standing open problem for the expansion of the basis-exchange graphs for matroids by \cite{ALOV19}.
However, these early results do not give non-trivial bounds unless the eigenvalues of the local walks are sufficiently small.
In contrast, \cite{AL20} obtain a great improvement where the bound is always meaningful as long as the local walks are not completely trivial,
which is currently the best result regarding the local-to-global principle for eigenvalues (or equivalently for variance contractions).

Despite these progresses, in many cases, 
there are inherent losses if one bounds the mixing times via eigenvalues or variance contractions \citep{LP17}.
A tighter relationship can be obtained via relative entropy contractions \citep{DS96}.
Although relative entropy decay is typically much harder to analyze than eigenvalues,
this method has been successfully applied to many problems, e.g.~\citep{JSTV04,Mor13}.
In the high-dimensional expander context,
\cite{CGM21} obtained optimal relative entropy contraction for the aforementioned matroid setting to get sharp mixing time bounds.

In this note we give a local-to-global principle for relative entropy contraction (\Cref{thm:main-ent}).
Our result is similar to that of \cite{AL20} for variance contraction, 
in the sense that our bounds are always non-trivial as long as the local walks are not completely trivial. 
In fact, our proof also works for variance contraction.
However, in that case our bounds are quantitatively no better than that of \cite{AL20}, 
because the local spectral gaps should satisfy the ``trickling-down'' theorem of \cite{Opp18} (see \Cref{sec:comparison} for some detailed comparison).
In some interesting special cases, our bound actually coincides with the counterpart of \cite{AL20} (see \Cref{cor:trickling-down-profile}).
In any case, our main interest lies in the local-to-global principle for relative entropy contraction,
which would lead to improved mixing time bounds in certain applications.

Our main result, \cref{thm:main-ent}, is independently obtained by \citet[Theorem 5.4]{CLV20}.
In fact, this theorem is one of the key ingredients of \cite{CLV20} to obtain optimal mixing times for Glauber dynamics over spin systems up to certain uniqueness conditions.
We refer the interested readers to their work for applications.
In this note we focus on the local-to-global principle for relative entropy contraction.

\section{Preliminaries}\label[section]{sec:prelim}
The main things that need to be defined in this section are simplicial complexes, as well as the distributions and random walks over them.
\subsection{Simplicial Complexes}

An (abstract) simplicial complex $\+C=(E,\+S)$ is a tuple of a ground set of elements $E$, and a nonempty downwards closed collection of sets $\+S$ (faces):
\begin{itemize}
    \item $\emptyset\in\+S$;
    \item if $S\in\+S$, $T\subseteq S$, then $T\in\+S$.
\end{itemize}
A pure simplicial complex has maximal sets of the same cardinality $d$. We denote by $\+C(k)$ the collection of sets of size $k$, where $0\le k\le d$.

We define the following distributions over $\+S$. First, a distribution $\pi_d$ is given on the top level, $\+C(d)$, and then lower level distributions defined as follows:
\begin{align*}
    \pi_k(S)\propto \sum_{T\in\+C(d) : T\supset S}\pi_d(T), && 0 \le k<d.
\end{align*}
Let $D_k:=\diag(\pi_k)$.
For a face $S\in\+S$, the link $\+C_S=(E\setminus S, \+S_S)$ is also a simplicial complex, 
where $\+S_S=\{T\mid T\subseteq E\setminus S, T\cup S\in\+S\}$.
We may similarly define distributions $\pi_{S,k}$ over $\+C_S(k)$ simply as
\begin{align*}
    \pi_{S,k}(T) \propto \pi_{\abs{S}+k}(T \cup S) , && 0 \le k \le d-\abs{S}.
\end{align*}
It is also convenient to equip $\+C$ with a weight function $w$, recursively defined as follows:
\begin{align*}
  w(S)\defeq 
  \begin{cases}
    \pi(S) & \text{if $\abs{S}=d$},\\
    \sum_{T\supset S,\; \abs{T}=\abs{S}+1} w(T) & \text{if $\abs{S}<d$}.
  \end{cases}  
\end{align*}
We will refer to the weight function of $\+C(k)$ as $w_k$, and, as before, the weight function of $\+C_S(k)$ will be denoted by $w_{S,k}$.

\subsection{Up and down random walks}

There are two natural exchange walks on $\+C(k)$, $\RWup{k}$ and $\RWdown{k}$, which start by adding or removing an element and coming back to $\+C(k)$. We call these walks ``global'' as they are defined over the whole of $\+C(k)$. They are comprised by the same two parts:
\begin{itemize}
    \item ``Going-up'', \Pup{k}; starting from a set $S\in\+C(k)$, we add an element $i\in E \setminus S$ with probability $\propto \pi_{k+1}(S \cup i)$.
    \item ``Going-down'', \Pdown{k}; starting from a set $S\in\+C(k)$, we remove an element $i\in S$ uniformly at random.
\end{itemize}
We can now write
\begin{align}
  \RWup{k} & = \Pup{k}\Pdown{k+1}, & % \label{eqn:up-decomp}\\
  \RWdown{k}& = \Pdown{k}\Pup{k-1}.   \label{eqn:down-decomp}
\end{align}
We shall also use the notation \RWup{S,k} and \RWdown{S,k}, where $S\in\+S$, to denote the walks on $\+C_S(k)$.
It is easy to check the detailed balance condition to see that \RWup{S,k} and \RWdown{S,k} are reversible with respect to $\pi_{S,k}$.

The ``local'' walks that will be particularly interesting are the walks $\RWup{S,1}$ and $\RWdown{S,2}$ for every face $S\in\+S$ with $\abs{S}\le d-2$. 
These walks are not completely local as they have weights that count higher level supersets. One can observe that the transition matrices of these walks are cospectral, because they are of the form $AB$ and $BA$.

The walks $G_S\defeq2\RWup{S,1}-I$ are the non-lazy version of the local walks $\RWup{S,1}$. 
A simplicial complex $\+C$ is called a $(a_0,...,a_{d-2})$-local-spectral expander if for any $S \in \+C(k)$, where $0 \le k\le d-2$,
\begin{align*}
  \lambda_2(G_S)\le a_{k},
\end{align*}
where $\lambda_2(\cdot)$ is the second largest eigenvalue.
We call the vector $(a_0,\dots,a_{d-2})$ satisfying the above a \emph{spectral profile} of $\+C$.

\begin{theorem}[Main Theorem of \citealp{AL20}]\label[theorem]{thm:AL}
    Let $\+C$ be a simplicial complex that is a $(a_0,...,a_{d-2})$-local-spectral expander. Then, for any $2\le k\le d$,
    \begin{align*}
      \lambda_2(\RWdown{k}) \le 1-\frac{1}{k}\prod_{i=0}^{k-2}\left(1-a_i\right).
    \end{align*}
  \end{theorem}
This is an example of a ``local-to-global'' theorem which is the type of theorem we are aiming for.

We define the \emph{Dirichlet form} of a reversible Markov chain $P$, over state space $\Omega$, as
\begin{align*}
  \Diri{P}{f,g} \defeq  \sum_{x,y\in\Omega}\pi(x)f(x)\big[\identity-P\big](x,y)g(y) = \transpose{f}diag(\pi)(I-P)g,
\end{align*}
where $f,g$ are two functions over $\Omega$, and $\pi$ is the stationary distribution of $P$.
Let the variance of $f$ be
\begin{align}\notag%\label{eqn:Var}
  \Var{\pi}{f}\defeq\Ex_{\pi}f^2-\left(\Ex_{\pi} f\right)^2.
\end{align}
The Poincar\'e inequality for $P$ is a variational way of characterizing the spectral gap:
\begin{align}\label{eqn:poincare}
  1-\lambda_2(P) = \inf\left\{\frac{\Diri{P}{f,f}}{\Var{\pi}{f}}\mid\; f:\Omega\rightarrow \=R\;,\;\Var{\pi}{f}\neq 0\right\}.
\end{align}

For $k\ge 2$ and a function $f^{(k)}:\+C(k)\rightarrow\=R$,
define $f^{(i)}:\+C(i)\rightarrow\=R$ for $1\le i\le k-1$ such that
\begin{align}  \label{eqn:f-level-i}
  f^{(i)}\defeq\prod_{j=i}^{k-1}\Pup{j} f^{(k)}.
\end{align}
For variances, we have
\begin{align}
  \Diri{\RWdown{k}}{f^{(k)},f^{(k)}} & = \Var{\pi_k}{f^{(k)}} - \Var{\pi_{k-1}}{f^{(k-1)}}; \label{eqn:var-equiv} \\
  \Diri{\RWup{k-1}}{f^{(k-1)},f^{(k-1)}} & = \Var{\pi_{k-1}}{f^{(k-1)}} - \Var{\pi_{k}}{\Pdown{k}f^{(k-1)}}. \label{eqn:var-equiv-up}
\end{align}
%There is also a similar equality if we premultiply by $\Pdown{k}$ to a function at level $k-1$:
%\begin{align*}
%    &=\transpose{\left(f^{(k-1)}\right)}\transpose{\left(\Pdown{k}\right)}D_k\Pdown{k}f^{(k-1)}- \left(\Ex_{\pi_k} \Pdown{k}f^{(k-1)}\right)^2\\
%    &=\transpose{\left(f^{(k-1)}\right)}D_{k-1}\RWup{k-1}f^{(k-1)}- \left(\Ex_{\pi_{k-1}} f^{(k-1)}\right)^2\\
%    &=-.
%\end{align*}
These equalities imply that the Poincar\'e inequalities for the walks $\RWdown{k}$ and $\RWup{k-1}$ are equivalent to variance contraction for the ``half'' walks $\Pdown{k}$ and $\Pup{k-1}$, respectively.

Moreover, let the (normalised) relative entropy of $f:\Omega\rightarrow\=R_{\ge 0}$ be
\begin{align*}
    \Ent{\pi}{f}\defeq\Ex_{\pi}(f\log f)-\Ex_{\pi}f\log \Ex_{\pi}f,
\end{align*}
where we follow the convention that $0\log 0=0$.
A related notion is the Kullback--Leibler (KL) divergence
%\begin{align*}
  $D(\tau~\|~\pi)\defeq\sum_{x\in\Omega}\tau(x)\log\left( \frac{\tau(x)}{\pi(x)} \right),$
%\end{align*}
where $\tau$ and $\pi$ are two distributions over the same $\Omega$.
Indeed, $D(\tau~\|~\pi)=\Ent{\pi}{f}$ where $f=\frac{\tau}{\pi}$.
The modified log-Sobolev constant \citep{BT06} is defined as
\begin{align}\label{eqn:mLSI}
  \rho(P)\defeq\inf\left\{\frac{\Diri{P}{f,\log f}}{\Ent{\pi}{f}}\mid\; f:\Omega\rightarrow \=R_{\ge 0}\;,\;\Ent{\pi}{f}\neq 0\right\}.
\end{align}
A related quantity is the relative entropy contraction ratio.
%For a reversible random walk $P$ over some state space $\Omega$ with stationary distribution $\pi$, 
%define the entropy contraction ratio:
%\begin{align}\label{eqn:entropy-contraction}
%  \beta(P)\defeq\sup\left\{\frac{\Ent{\pi}{Pg}}{\Ent{\pi}{g}}~\middle |~ g:\Omega\rightarrow\=R_{\ge 0},~\Ent{\pi}{g}\neq 0 \right\}.
%\end{align}
%Note that $0\le \beta(P)\le 1$.
%It is easy to show that $\rho(P)\ge 1-\beta(P)$.
%In fact, $\rho(P)$ is the derivative (against time) of relative entropy for the continuous-time Markov chain related to $P$ \citep{BT06}.

%For entropy we get inequalities instead of equalities as is the case for variance (\ref{eqn:var-equiv}) (We assume $\Ex_{\pi_k} f = 1$):
For entropies and $f^{(k)}:\+C(k)\rightarrow\=R_{\ge 0}$, we get inequalities instead of equalities as is the case for variances, \eqref{eqn:var-equiv} and \eqref{eqn:var-equiv-up},
\begin{align}
  \Diri{\RWdown{k}}{f^{(k)},\log f^{(k)}} & \ge \Ent{\pi_k}{f^{(k)}} - \Ent{\pi_{k-1}}{f^{(k-1)}}; \label{eqn:entropy-ineq}\\
  \Diri{\RWup{k-1}}{f^{(k-1)},\log f^{(k-1)}} & \ge \Ent{\pi_{k-1}}{f^{(k-1)}} - \Ent{\pi_{k}}{\Pdown{k}f^{(k-1)}}\label{eqn:entropy-ineq-up}.
\end{align}
%and
%\begin{align*}
%    \Ent{\pi_{k}}{\Pdown{k}f^{(k-1)}}&=\transpose{\left(f^{(k-1)}\right)}\transpose{\left(\Pdown{k}\right)}D_k\log \Pdown{k}f^{(k-1)}\\
%    &\ge \transpose{\left(f^{(k-1)}\right)}\transpose{\left(\Pdown{k}\right)}D_k\Pdown{k}\log f^{(k-1)} \tag{by Jensen's inequality}\\
%    &=\transpose{\left(f^{(k-1)}\right)}D_{k-1}\RWup{k-1}\log f^{(k-1)}\\
%    &=\Ent{\pi_{k-1}}{f^{(k-1)}}-\Diri{\RWup{k-1}}{f^{(k-1)},\log f^{(k-1)}}.
%\end{align*}
Thus, entropy contractions of $\Pdown{k}$ and $\Pup{k-1}$ imply modified log-Sobolev inequalities of $\RWdown{k}$ and $\RWup{k-1}$, respectively,
but not the other way around, unlike the variance case.
%According to this, we get a modified log-Sobolev inequality from entropy contraction but not the other way round. 

One important result regarding the spectral profile of a simplicial complex is the trickling down theorem due to \cite{Opp18}.
Here we give an alternative variational proof.
\begin{theorem}[Trickling down theorem, \citealp{Opp18}]\label[theorem]{thm:trickling-down}
    Let $\+C$ be a simplicial complex that is a $\gamma$-local-spectral expander at level $k$, meaning that $\lambda_2(G_T)\le \gamma <1$ for all $T \in \+C(k)$. 
    Then, if $\lambda_2(G_S)<1$ (connectedness) for some $S \in \+C(k-1)$,
    \begin{align*}
        \lambda_2(G_S) \le \frac{\gamma}{1-\gamma}.
    \end{align*}
\end{theorem}
\begin{proof}
It suffices to prove the theorem for $k=1$. For any face $S$, let $D_{S,1} \defeq diag(\pi_{S,1})$, and note the decomposition
\begin{align} \label{eqn:up-decomposition-D}
    D_1 = \sum_{v \in \+C(1)}\pi_1(v)\overline{D_{v,1}},
\end{align}
where the overline denotes the extension by zeros to the appropriate dimension. Also, for any face $S$, $G_S = diag(w_{S,1})^{-1}W_S$, where $W_S(u,v) \defeq w_{S,2}(\{u,v\})$. We also have the following decomposition,
\begin{align*}
  G_\emptyset = diag(w_1)^{-1}W_\emptyset &= \frac{1}{w(\emptyset)}D_1^{-1}\sum_{v \in \+C(1)}\overline{W_v} = D_1^{-1}\sum_{v \in \+C(1)}\frac{\overline{diag(w_{v,1})G_v}}{w(\emptyset)},
\end{align*}
or equivalently,
\begin{align} \label{eqn:up-decomposition-G}
  G_\emptyset &= D_1^{-1}\sum_{v \in \+C(1)}\pi_1(v)\overline{D_{v,1}}\overline{G_v}.
\end{align}
Moreover, notice that $G_{\emptyset}(v,u) = \pi_{v,1}(u)$ and for any function $f:\+C(1)\rightarrow\=R$,
\begin{align}\label{eqn:G-and-pi_v}
  G_\emptyset f(v) = \sum_{u\in\+C(1)}G_{\emptyset}(v,u)f(u) = \overline{\pi_{v,1}}f.
\end{align}
Given these, for any function $f:\+C(1)\rightarrow\=R$, we can write
\begin{align*}
    \Diri{G_\emptyset}{f,f} &= \transpose{f}D_1\left(I-G_\emptyset\right)f\\
    &= \sum_{v \in \+C(1)}\pi_1(v)\left[\transpose{f}\overline{D_{v,1}}\left(I-\overline{G_v}\right) f\right]\tag{by \eqref{eqn:up-decomposition-D} and \eqref{eqn:up-decomposition-G}}\\
    &= \sum_{v \in \+C(1)}\pi_1(v)\Diri{G_v}{f_v,f_v} \tag{where $f_v$ is $f$ restricted to $\+C_v(1)$}\\
    &\ge (1-\gamma)\sum_{v \in \+C(1)}\pi_1(v)\Var{\pi_{v,1}}{f_v}\tag{because $\lambda_2(G_v)\le \gamma$}\\
    &= (1-\gamma)\sum_{v \in \+C(1)}\pi_1(v)\left[\transpose{f}\overline{D_{v,1}}f-(\overline{\pi_{v,1}}f)^2\right]\\
    &= (1-\gamma)\left[\Var{\pi_1}{f}+(\pi_1 f)^2-\sum_{v \in \+C(1)}\pi_1(v)\left([G_\emptyset f](v)\right)^2\right]\tag{by \eqref{eqn:up-decomposition-D} and \eqref{eqn:G-and-pi_v}}\\
    &= (1-\gamma)\left[\Var{\pi_1}{f}-\Var{\pi_1}{G_\emptyset f}\right].
\end{align*}
Since the inequality above holds for any $f$, we can choose $f=v_i$, where $v_i$ is an eigenvector of $G_\emptyset$ corresponding to eigenvalue $\lambda_i$. Then,
\begin{align*}
  \Diri{G_\emptyset}{v_i,v_i} &\ge (1-\gamma)\left[\Var{\pi_1}{v_i}-\Var{\pi_1}{G_\emptyset v_i}\right],
\end{align*}
which simplifies into
\begin{align*}
  (1-\lambda_i)\transpose{v_i}D_1 v_i &\ge (1-\gamma)(1-\lambda_i^2)\transpose{v_i}D_1 v_i.
\end{align*}
Thus, $(1-\lambda_i) \ge (1-\gamma)(1-\lambda_i^2)$.
In particular, if $\lambda_2(G_\emptyset)<1$, then
\begin{align*}
    \lambda_2(G_\emptyset) &\le \frac{\gamma}{1-\gamma}.\qedhere
\end{align*}
\end{proof}

An interesting question is to find a similar theorem for entropies or modified log-Sobolev constants.
However, a straightforward generalisation of the proof above would involve the ratio between $\Ent{\pi_1}{G_{\emptyset}f}$ and $\Ent{\pi_1}{f}$.
For variances, the ratio between $\Var{\pi_1}{G_{\emptyset}f}$ and $\Var{\pi_1}{f}$ can be related to $\lambda_2$, as shown in the proof.
For relative entropy, that ratio does not seem to be directly related to $\rho(G_{\emptyset})$.
In fact, $\rho$ is the change rate of the relative entropy for the continuous-time Markov chain \citep{BT06}, 
and can be related to the change rate for the discrete-time chain if there is no negative eigenvalue, as shown by \cite{Mic97}.
However, this is not the case for $G_{\emptyset}$.
It remains open to find relative entropy contraction descent similar to \Cref{thm:trickling-down}.

\section{Decay of variance}

We first study the local-to-global principle for variances.
We show a theorem similar to the main result of \cite{AL20}.
The bound we obtain will be given by a recursion, and seems to be incomparable to that of \cite{AL20} at first sight.
However, it turns out that our bound is weaker for spectral profiles satisfying the trickling down theorem,
and it coincides with the bound of \cite{AL20} when the trickling down theorem is tight.
Nonetheless, our proof approach is different from that of \cite{AL20},
and our approach has the advantage that it generalises to the entropy case as shown in \Cref{sec:entropy}.

We use the following decomposition, which was shown by \cite{CGM21}.
For completeness we give a proof here.
\begin{lemma}  \label[lemma]{lem:decomposition}
  Let $k\ge 2$ and $f^{(k)}:\+C(k)\rightarrow \=R$ be a function on $\+C(k)$.
  Then,
  \begin{align*}
    \Var{\pi_{k}}{f^{(k)}} &= \sum_{S\in\+C(k-2)}\pi_{k-2}(S)\Var{\pi_{S,2}}{f_S^{(2)}} + \Var{\pi_{k-2}}{f^{(k-2)}},
  \end{align*}
  where $f_S^{(2)}(T)\defeq f^{(k)}(S\cup T)$ for $T\in\+C_S(2)$.
  Moreover, the same decomposition holds for $f^{(k)}:\+C(k)\rightarrow \=R_{\ge 0}$ with $\Var{}{}$ replaced by $\Ent{}{}$.
\end{lemma}
\begin{proof}
  We may assume that $\Ex_{\pi_k}f^{(k)}=0$.
  By direct calculation, for any $I\in \+C(k)$,
  \begin{align*}
    \pi_k(I)=\sum_{S\in\+C(k-2), S\subset I}\pi_{k-2}(S)\pi_{S,2}(I\setminus S).
  \end{align*}
  Notice that for any $S\in\+C(k-2)$, $\Ex_{\pi_{k-2}} f^{(k-2)}=\Ex_{\pi_k}f^{(k)}=0$, and 
  \begin{align*}
    f^{(k-2)}(S) & = \sum_{T\in\+C_S(2)}\frac{2w(S\cup T)}{w(S)}f^{(k)}(S\cup T) = \Ex_{\pi_{S,2}}f_S^{(2)}.
  \end{align*}
  Thus we have
  \begin{align*}
    \Var{\pi_{k}}{f^{(k)}} & = \sum_{S\in\+C(k-2)}\pi_{k-2}(S)\Var{\pi_{S,2}}{f_S^{(2)}} + \sum_{S\in\+C(k-2)}\pi_{k-2}(S)\left(\Ex_{\pi_{S,2}}f_S^{(2)}\right)^2\\
    & = \sum_{S\in\+C(k-2)}\pi_{k-2}(S)\Var{\pi_{S,2}}{f_S^{(2)}} + \Var{\pi_{k-2}}{f^{(k-2)}}.
  \end{align*}
  The proof for $\Ent{}{}$ is completely analogous.
\end{proof}

Now we are ready to show the local-to-global principle for variances.

\begin{theorem}  \label[theorem]{thm:main-var}
    Let $\+C$ be a simplicial complex that is a $(a_0,...,a_{d-2})$-local-spectral expander, and let $s_k \defeq \frac{2}{1+a_k}$. Then, for any $2\le k\le d$,
    \begin{align*}
        \lambda_2(\RWdown{k})=\lambda_2(\RWup{k-1}) \le \frac{1}{v_{k-2}},
    \end{align*}
    where $v_k$ is recursively defined as
    \begin{align}\label{eqn:recursion}
        v_k = s_k - \frac{s_k-1}{v_{k-1}}, && v_0 = s_0.
    \end{align}
\end{theorem}
\begin{proof}
    From the assumption on local expansion we get that for any $S \in \+C(k-2)$, where $2 \le k\le d$,
    \begin{align*}
      \lambda_2(G_S)\le a_{k-2} \implies \lambda_2(\RWdown{S,2})\le \frac{1+a_{k-2}}{2}.
    \end{align*}
    Then, the spectral gap $\lambda(\RWdown{S,2})=1-\lambda_2(\RWdown{S,2})\geq \frac{1-a_{k-2}}{2}$.
    By \eqref{eqn:poincare},
    this implies that for any $g:\+C_S(2)\rightarrow \=R$,
    \begin{align*}
      \Diri{\RWdown{S,2}}{g,g}\ge \left(\frac{1-a_{k-2}}{2}\right)\Var{\pi_{S,2}}{g}.
    \end{align*}
    Remembering that $s_{k-2} = \frac{2}{1+a_{k-2}}$, and \eqref{eqn:var-equiv}, we obtain
    \begin{align} \label{eqn:links}
      \Var{\pi_{S,2}}{g} \ge s_{k-2}\Var{\pi_{S,1}}{\Pup{S,1}g}.
    \end{align}
%    For $S=\emptyset\in \+C(0)$ we have a statement for the global variances on levels 1 and 2,
%    \begin{align} \label{eqn:empty}
%        \Var{\pi_2}{f^{(2)}} \ge s_0\Var{\pi_{1}}{f^{(1)}}.
%    \end{align}
    
    % For $T\in\+C(k)$, $S\in\+C(i)$ and $S\subset T$,
    % \begin{align*}
    %     \pi_{k}(T) = \frac{w(T)}{Z_{k}} = \frac{k!}{i!(k-i!)}\frac{w(S)}{Z_{i}}\cdot \frac{(k-i)!w(T)}{w(S)} = \binom{k}{i}\pi_{i}(S)\pi_{S,k-i}(T),
    % \end{align*}
    % as $i!Z_{i}=k!Z_{k}$.
    % This means that
    % \begin{align}\label{eqn:decomp}
    %     \pi_{k}(T) =  \sum_{S\in\+C(i),S\subset T}\pi_{i}(S)\pi_{S,k-i}(T) = \sum_{S\in\+C(i)}\pi_{i}(S)\pi_{S,k-i}(T).
    % \end{align}
    
    We finish the proof by an induction on $k$.
    The base case of $k=2$ is straightforward by noticing that $v_0=s_0$.
    %Given these, we can deduce an inequality between the global variances at levels $k$ and $k-1$, for $k\ge 3$, 
    For the induction step on $k\ge 3$, by \Cref{lem:decomposition}, we have
    \begin{align}
        \Var{\pi_{k}}{f^{(k)}} &= \sum_{S\in\+C(k-2)}\pi_{k-2}(S)\Var{\pi_{S,2}}{f_S^{(2)}} + \Var{\pi_{k-2}}{f^{(k-2)}} \notag\\
        &\ge  s_{k-2}\sum_{S\in\+C(k-2)}\pi_{k-2}(S)\Var{\pi_{S,1}}{\Pup{S,1}f_S^{(2)}} + \Var{\pi_{k-2}}{f^{(k-2)}} \tag{by \eqref{eqn:links}}\\
        &=  (1-\epsilon)s_{k-2}\sum_{S\in\+C(k-2)}\pi_{k-2}(S)\Var{\pi_{S,1}}{\Pup{S,1}f_S^{(2)}} \notag\\
        &+\epsilon s_{k-2}\sum_{S\in\+C(k-2)}\pi_{k-2}(S)\Var{\pi_{S,1}}{\Pup{S,1}f_S^{(2)}}+ \Var{\pi_{k-2}}{f^{(k-2)}} \label{eqn:decomp-k},
    \end{align}
    where $\epsilon\ge 0$ will be chosen later.
    A similar decomposition holds for level $k-1$,
    \begin{align*}
        \Var{\pi_{k-1}}{f^{(k-1)}}=\sum_{S\in\+C(k-2)}\pi_{k-2}(S)\Var{\pi_{S,1}}{f_S^{(1)}} + \Var{\pi_{k-2}}{f^{(k-2)}},
    \end{align*}
    where $f_S^{(1)}(e)\defeq f^{(k-1)}(S\cup \{e\})$ for $e\in\+C_S(1)$.
    Together with the induction hypothesis this implies that
    \begin{align}\label{eqn:var-IH}
        \sum_{S\in\+C(k-2)}\pi_{k-2}(S)\Var{\pi_{S,1}}{f_S^{(1)}} \ge (v_{k-3}-1)\Var{\pi_{k-2}}{f^{(k-2)}}.
    \end{align}
    Finally, notice that $\Pup{S,1}f_S^{(2)} = f_S^{(1)}$.
    Plugging \eqref{eqn:var-IH} into \eqref{eqn:decomp-k} and choosing $\epsilon = \frac{s_{k-2}-1}{v_{k-3}s_{k-2}}$,
%     \begin{align*}
%        &\Var{\pi_{k}}{f^{(k)}}\\
%        &\ge  (1-\epsilon)s_{k-2}\sum_{S\in\+C(k-2)}\pi_{k-2}(S)\Var{\pi_{S,1}}{f^{(k-1)}} +\left[\epsilon s_{k-2}(v_{k-3}-1)+1\right]\Var{\pi_{k-2}}{f^{(k-2)}}.
%    \end{align*}
%    Now, we need to pick $\epsilon$ such that
%    \begin{align*}
%        (1-\epsilon)s_{k-2}&=\epsilon s_{k-2}(v_{k-3}-1)+1,
%    \end{align*}
%    namely $\epsilon = \frac{s_{k-2}-1}{v_{k-3}s_{k-2}}$.
%    With such an $\epsilon$ we have that
    \begin{align*}
        \Var{\pi_{k}}{f^{(k)}} &\ge  (1-\epsilon)s_{k-2}\Var{\pi_{k-1}}{f^{(k-1)}}\\
        &=  \left(s_{k-2}-\frac{s_{k-2}-1}{v_{k-3}}\right)\Var{\pi_{2}}{f^{(k-1)}} = v_{k-2}\Var{\pi_{k-1}}{f^{(k-1)}}.
    \end{align*}
    By \eqref{eqn:poincare} and \eqref{eqn:var-equiv}, this translates to the spectral gap bound
%    \begin{align*}
    $1-\lambda_2(\RWdown{k}) \ge 1-\frac{1}{v_{k-2}}$,
%    \end{align*}
    namely, $\lambda_2(\RWdown{k}) \le \frac{1}{v_{k-2}}$.
\end{proof}

A particular interesting case is when $a_{k-2} = \frac{\gamma}{1-(d-k)\gamma}$, for any $\gamma \le 1/(d-1)$.
%Our approach yields the same bound as \cite{AL20} whenever  
This profile of $a$'s is exactly the result of repeatedly applying the trickling down theorem, starting from level $d-2$ with $a_{d-2} = \gamma$, down to level $0$.
%Let us prove this result for completeness.
\begin{corollary}\label[corollary]{cor:trickling-down-profile}
  Let $\mathcal{C}$ be a pure simplicial complex of dimension $d$ and suppose that $a_{d-2}\le \gamma \le 1/(d-1)$. 
  Then, for any $2 \le k \le d$,
  \begin{align*}
      \lambda_2(\RWdown{k}) = \lambda_2(\RWup{k-1}) \le 1-\frac{1}{k}\left(\frac{1-(d-1)\gamma}{1-(d-k)\gamma}\right).
  \end{align*}
\end{corollary}
\begin{proof}
  By the trickling down theorem (\cref{thm:trickling-down}), we immediately get that $a_{k-2} \le \frac{\gamma}{1-(d-k)\gamma}$. 
  Due to \cref{thm:main-var}, it suffices to prove that $\frac{1}{v_{k-2}} \le 1-\frac{1}{k}\left(\frac{1-(d-1)\gamma}{1-(d-k)\gamma}\right)$.
  For $k=2$, the inequality holds trivially by the definition of $v_0$. 
  Suppose that the inequality holds up to dimension $k-1 \ge 2$. Then,
  \begin{align*}
      \frac{1}{v_{k-3}} \le 1-\frac{1}{k-1}\left(\frac{1-(d-1)\gamma}{1-(d-k+1)\gamma}\right),
  \end{align*}
  and as $a_{k-2} \le \frac{\gamma}{1-(d-k)\gamma}$,
  \begin{align*}
    s_{k-2} = \frac{2}{a_{k-2}+1}\ge \frac{2(1-(d-k)\gamma)}{1-(d-k-1)\gamma}.
  \end{align*}
  As in \cref{thm:main-var},
    \begin{align*}
        v_{k-2} &= s_{k-2} - \frac{s_{k-2}-1}{v_{k-3}}\\
        &\ge s_{k-2} - \left(s_{k-2}-1\right)\left[1-\frac{1}{k-1}\left(\frac{1-(d-1)\gamma}{1-(d-k+1)\gamma}\right)\right]\\
        &= 1+(s_{k-2}-1)\frac{1}{k-1}\left(\frac{1-(d-1)\gamma}{1-(d-k+1)\gamma}\right)\\
%        &= 1+\left(\frac{2}{1+a_{k-2}}-1\right)\frac{1}{k-1}\left(\frac{1-(d-1)\gamma}{1-(d-k+1)\gamma}\right)\\
        %&\ge 1+\left(\frac{2}{1+\frac{\gamma}{1-(d-k)\gamma}}-1\right)\frac{1}{k-1}\left(\frac{1-(d-1)\gamma}{1-(d-k+1)\gamma}\right) \tag{$a_{k-2} \le \frac{\gamma}{1-(d-k)\gamma}$ and $\gamma \le \frac{1}{d-1}$}\\
        & \ge 1+\left(\frac{1-(d-k+1)\gamma}{1-(d-k-1)\gamma}\right)\frac{1}{k-1}\left(\frac{1-(d-1)\gamma}{1-(d-k+1)\gamma}\right)\\
        &= 1+\frac{1}{k-1}\left(\frac{1-(d-1)\gamma}{1-(d-k-1)\gamma}\right),%\\
%        &= \frac{(k-1)[1-(d-k-1)\gamma]+1-(d-1)\gamma}{(k-1)[1-(d-k-1)\gamma]},\\
    \end{align*}
    the reciprocal of which is
    \begin{align*}
        \frac{1}{v_{k-2}} %&\le \frac{(k-1)[1-(d-k-1)\gamma]}{(k-1)[1-(d-k-1)\gamma]+1-(d-1)\gamma}\\
%        &=1-\frac{1-(d-1)\gamma}{(k-1)[1-(d-k-1)\gamma]+1-(d-1)\gamma}\\
        & \le 1-\frac{1}{k}\left(\frac{1-(d-1)\gamma}{1-(d-k)\gamma}\right).
    \end{align*}
    The corollary follows by induction together with \cref{thm:main-var}.
\end{proof}

In particular, by setting $\gamma=1/d$, this retrieves \citet[Corollary 1.6]{AL20}. 

\begin{corollary}[\protect{\citealp[Corollary 1.6]{AL20}}]\label[corollary]{cor:1/k}
  If $\mathcal{C}$ is a $\frac{1}{d}$-local-spectral expander at level $d-2$, i.e. $a_{d-2} \le 1/d$,
    then for any $2\le k\le d$,
    \begin{align*}
        \lambda_2(\RWdown{k}) \le 1-\frac{1}{k^2}.
    \end{align*}
\end{corollary}

\Cref{cor:trickling-down-profile} and \Cref{cor:1/k} are alternative derivations of a useful result of \cite{AL20}, which requires minimal assumptions on the expansion of the local graphs nearest to the top level. If we are studying the uniform distribution on the top level, these graphs are unweighted, as opposed to the lower level graphs that have weights that count some difficult to analyze and possibly intractable quantities.

An important advantage of our method is that it can also be used when we assume entropy contraction factors on the local walks.
However, we do not know a counterpart to the trickling down theorem for entropies.

The contraction ratio $v_k$ is given by the recursion in \eqref{eqn:recursion}.
Next we will solve the recursion to give an explicit expression.
Do the following substitution $x_k \defeq \frac{v_k}{v_k-1}$ and notice $x_0=\frac{v_0}{v_0-1}=1+\frac{1+a_0}{1-a_0}$.
The recurrence simplifies into
\begin{align}\label{eqn:recurrence}
  x_k = \frac{1+a_k}{1-a_k}\cdot x_{k-1}+1,
\end{align}
and can be solved $x_k = \sum_{i=0}^{k} S_i^k + 1$, where $S_i^k \defeq \prod_{j=i}^k \frac{1+a_j}{1-a_j} = \prod_{j=i}^k \frac{1}{s_j-1}$.
Then,
\begin{align}   \label{eqn:v_k}
  v_k=\frac{x_k}{x_k-1} = 1+\frac{1}{\sum_{i=0}^k S_i^k}.
\end{align}
Our second largest eigenvalue bound is
\begin{align}\label{eqn:our-bound}
  \gamma_k = \frac{1}{v_{k-2}} = 1 - \frac{1}{\sum_{i=0}^{k-2} S_i^{k-2} + 1}.
\end{align}

%Lastly, %\htodo{This paragraph of reflection is nice, but I think it's not necessary to have it in the final paper. It is certainly fine for the review document.} 
%let us mention that one can construct other recursions based on the ideas above, which however seem to be equivalent.
%For example, if we always decompose at level $1$ instead of $k-2$, we get a different recursion that also seems to coincide with \cite{AL20} for $a_i = \frac{1}{2+S}$: Let $v_{k-2}^l$ symbolize the factor between the general variances at level $k$ and $k-1$ after conditioning on $l$ elements. ($v_0^0=s_0$, $v_1^0=\frac{s_0s_{1}-s_{1}+1}{s_0}$, $v_1^1=\frac{s_{1}s_2-s_2+1}{s_{1}}$.)
%Then,
%\begin{align*}
%    v_{k}^0 = \frac{v_{k-1}^1\prod_{i=0}^{k-1}v_i^0-v_{k-1}^1+1}{\prod_{i=0}^{k-1}v_i^0}.
%\end{align*}

\subsection{Comparison with \texorpdfstring{\cite{AL20}}{Alev and Lau 2020}}
\label[section]{sec:comparison}

Given a spectral profile $(a_0,\dots,a_{d-2})$,
recall our second largest eigenvalue bound \eqref{eqn:our-bound}.
In contrast, \Cref{thm:AL} \citep[Theorem 1.5]{AL20} achieves a different upper bound
\begin{align}  \label{eqn:AL-bound}
  \gamma_{k,AL} : = 1-\frac{1}{k}\prod_{i=0}^{k-2}(1-a_i).
\end{align}
We call a spectral profile \emph{admissible} if 
\begin{enumerate}
  \item for all $0\le i\le d-2$, $a_i<1$;
  \item for all $1\le i\le d-2$, $a_{i-1}\le \frac{a_i}{1-a_i}$.
\end{enumerate}
Note that the first condition here ensures that the random walk over the links are all connected,
and the second condition ensures that the spectral profile is consistent with the trickling down theorem, \Cref{thm:trickling-down}.

Our bound in \eqref{eqn:our-bound} is no better than the bound of \eqref{eqn:AL-bound} by \cite{AL20} when $(a_0,\dots,a_{d-2})$ is admissible.

\begin{proposition}  \label[proposition]{prop:comparison}
  Let $(a_0,\dots,a_{d-2})$ be an admissible spectral profile.
  For any $0\le k\le d-2$,
  \begin{align*}
    \gamma_k\ge \gamma_{k,AL}.
  \end{align*}
\end{proposition}

To show \Cref{prop:comparison}, we need to first show a lemma.

\begin{lemma}  \label[lemma]{lem:spectral-profile-property}
  Let $(a_0,\dots,a_{d-2})$ be an admissible spectral profile.
  For any $1\le k\le d-2$,
  \begin{align*}
    a_{k}(k+1)+\prod_{i=0}^{k}(1-a_i)\ge 1.
  \end{align*}
\end{lemma}
\begin{proof}
  We do an induction on $k$.
  For $k=1$, we have 
  \begin{align*}
    2a_1+(1-a_0)(1-a_1) & = 1+(a_1+a_0a_1-a_0) \ge 1,
  \end{align*}
  where the inequality is due to $a_0\le \frac{a_1}{1-a_1}$.
  For the induction step, suppose that the lemma holds for some $k\ge 1$,
  namely,
%  \begin{align*}
    $\prod_{i=0}^{k}(1-a_i) \ge 1 - a_k(k+1).$
%  \end{align*}
  Thus,
  \begin{align*}
    a_{k+1}(k+2)+\prod_{i=0}^{k+1}(1-a_i) & \ge  a_{k+1}(k+2)+(1-a_{k+1})(1 - a_k(k+1)) \\
%    & = a_{k+1}(k+2)+ 1 - a_{k+1} - a_k(k+1) + a_ka_{k+1}(k+1)\\
    & = 1 + (k+1)(a_{k+1}+a_ka_{k+1}-a_k) \ge 1,
  \end{align*}
  where the last inequality is because $a_k\le \frac{a_{k+1}}{1-a_{k+1}}$.
\end{proof}

With \Cref{lem:spectral-profile-property}, we can now prove \Cref{prop:comparison}.

\begin{proof}[Proof of \Cref{prop:comparison}]
  All we need to show is that for all $0\le k\le d-2$,
  \begin{align*}
    x_{k} \ge \frac{k+2}{\prod_{i=0}^{k}(1-a_i)}.
  \end{align*}
  We do an induction on $k$.
  For $k=0$, $x_0=\frac{2}{1-a_0}$ and the claim holds.
  
  By \eqref{eqn:recurrence}, we have that
  \begin{align*}
    x_k & = \frac{1+a_k}{1-a_k} x_{k-1}+1 \\
    & \ge \frac{(1+a_k)(k+1)}{\prod_{i=0}^{k}(1-a_i)}+1 = \frac{(1+a_k)(k+1)+\prod_{i=0}^{k}(1-a_i)}{\prod_{i=0}^{k}(1-a_i)} \tag*{(by induction hypothesis)}\\
    & \ge \frac{k+2}{\prod_{i=0}^{k}(1-a_i)} \tag*{(by \protect\Cref{lem:spectral-profile-property})}.
  \end{align*}
  This finishes the induction.
\end{proof}

Notice that the equality in \Cref{prop:comparison} holds when the trickling down theorem is tight.
In this case our bound coincides with \cite{AL20}, which is shown in \Cref{cor:trickling-down-profile}.

%To get a concrete comparison,
%consider the case where $a_i = a<1$ for all $2\le i\le d$.
%Then $S^k_i = \beta^{k-i+1}$,
%where $\beta = \frac{1+a}{1-a}>1$, and $v_0=s_0=\frac{2}{1+a}$, $x_0 = \frac{v_0}{v_0-1} = 2\beta$.
%We have
%\begin{align*}
%  \gamma_k &= 1 - \frac{1}{2\beta^{k+1}+\sum_{i=0}^{k-1}\beta^{i}}= 1 - \frac{1}{2\beta^{k+1} + \frac{\beta^k-1}{\beta-1}}.
%\end{align*}
%On the other hand,
%\begin{align*}
%  \gamma_{k,AL} &\le 1-\frac{1}{k}(1-a)^{k-1}.
%\end{align*}
%More specifically, suppose $a=1/2$, similar to the construction of \cite{LMY20}.
%In this case, $\gamma_{k,AL}=1-\Omega(k^{-1}2^{-k})$ whereas $\gamma_k = 1- \Omega(3^{-k})$.

\section{Decay of relative entropy} \label[section]{sec:entropy}

In this section we obtain a local-to-global principle for entropy contractions.
We obtain bounds for the relative entropy decay of each step of $\RWdown{k}$, based solely on a property of the local walks on the faces. 
This answers a question raised by \cite{alev2020higher}: ``\dots{} is there a property of the local graphs $G_{\alpha}$ that would allow us to bound the (modified) log-Sobolev constant of the corresponding chain as opposed to the spectral gap?''
%``More generally, is there a property of the local graphs $G_{\alpha}$ that would allow us to bound the (modified) log-Sobolev constant of the corresponding chain as opposed to the spectral gap? This would give a more systematic way of studying the (modified) log-Sobolev constants of many interesting Markov chains, which in practice is known to be a challenging task.''

The idea is to follow the argument of \cref{thm:main-var} but with every occurrence of Var$()$ replaced with Ent$()$. 
Our initial assumptions for the local walks would now be entropy contraction, and the recursion would be the same. Let us state the theorem for clarity.

\begin{theorem}  \label[theorem]{thm:main-ent}
    Let $\+C$ be a simplicial complex that satisfies the local inequalities
    \begin{align*}
      \Ent{\pi_{S,2}}{f_S^{(2)}} \ge s_{k-2}\Ent{\pi_{S,1}}{f_S^{(1)}},
    \end{align*}
    for any $2\le k\le d$, $S \in \+C(k-2)$, and $f_S^{(2)}:\+C_S(2)\rightarrow \=R_{\ge 0}$,
    where $\{s_k\}$ are some entropy contraction factors greater than or equal to $1$. 
    Then, for any $2\le k\le d$ and $f^{(k)}:\+C(k)\rightarrow \=R_{\ge 0}$, we get the global inequalities
    \begin{align*}
        \Ent{\pi_{k}}{f^{(k)}} &\ge  v_{k-2}\Ent{\pi_{k-1}}{f^{(k-1)}},
    \end{align*}
    where $v_k$ is recursively defined as in \eqref{eqn:recursion}.
%    \begin{align*}
%        v_k = s_k - \frac{s_k-1}{v_{k-1}}, && v_0 = s_0.
%    \end{align*}
\end{theorem}
\begin{proof}
   The proof is identical to that of \cref{thm:main-var} except that Var$()$ is replaced by Ent$()$.
\end{proof}

Note that an explicit formula for $v_k$ is given in \eqref{eqn:v_k}.

As entropy contraction implies the modified log-Sobolev inequality,
we have the following corollary.

\begin{corollary}
    Let $\mathcal{C}$ be a simplicial complex that satisfies the assumptions of \cref{thm:main-ent}. Then, the following hold:
    \begin{itemize}
        \item for any $2\le k\le r$,
          $\rho(\RWdown{k})\ge1-\frac{1}{v_{k-2}}$;
        \item for any $1\le k\le r-1$,
         $\rho(\RWup{k})\ge1-\frac{1}{v_{k-1}}$. 
    \end{itemize}
\end{corollary}
\begin{proof}
  For the down-up walk, combine \Cref{thm:main-ent} with \eqref{eqn:entropy-ineq}. For a proof that applies to both walks, see \cite{CGM21}.
\end{proof}

We remark that our \Cref{thm:main-ent} is independently obtained by \citet[Theorem 5.4]{CLV20}.
For applications of \Cref{thm:main-ent}, we refer the reader to \cite{CLV20}.
Although the conditions of \Cref{thm:main-ent} seem restrictive,
it is in fact satisfied by certain high-dimensional expanders with bounded marginals, 
as shown in \cite{CLV20}.

%\subsection{Comparison with \texorpdfstring{\cite{CLV20}}{Chen et al.~2020}}
%
%We remark that our \Cref{thm:main-ent} is the same as \citet[Theorem 5.4]{CLV20}.
%The notations are different, so we give some detail here.
%
%First, the quantity $\alpha_k$ in \cite{CLV20} equals our $s_k-1=\frac{1-a_k}{1+a_k}$.
%Thus, 
%\begin{align*}
%  x_0=\frac{s_0}{s_0-1}=1+\frac{1}{\alpha_0},
%\end{align*}
%and $S_i^k$ defined in \eqref{eqn:S_i^k} is
%\begin{align*}
%  S_i^k & =\prod_{j=i}^k\frac{1+a_k}{1-a_k} = \prod_{j=i}^k \frac{1}{\alpha_j}.
%\end{align*}
%In particular, this means \eqref{eqn:x_k} can be simplified into
%\begin{align*}
%  x_k=\sum_{i=0}^k S_i^k+1.
%\end{align*}
%Thus,
%\begin{align*}
%  v_k & =\frac{x_k}{x_k-1} = \frac{\sum_{i=0}^k S_i^k+1}{\sum_{i=0}^k S_i^k}
%  = \frac{\sum_{i=0}^k \prod_{j=i}^k \frac{1}{\alpha_j}+1}{\sum_{i=0}^k \prod_{j=i}^k \frac{1}{\alpha_j}}\\
%  & = \frac{1+\sum_{i=1}^k \prod_{j=0}^{i-1} {\alpha_j}}{\sum_{i=1}^{k-1} \prod_{j=0}^{i-1} \alpha_j}, \tag{by multiplying $\prod_{j=0}^k \alpha_j$}
%\end{align*}
%which is the contraction ratio in \citet[Theorem 5.4]{CLV20}.

\bibliographystyle{plainnat}
\bibliography{Log-Sob}

%\listoftodos

\end{document}